\documentclass[runningheads]{llncs}
\usepackage[T1]{fontenc}
\usepackage{graphicx}
\usepackage{algorithm}
\usepackage{algorithmic}
\usepackage{bm}
\usepackage{amsmath}
\usepackage{amsfonts}
\usepackage{amssymb}
\usepackage{multirow}
\usepackage{adjustbox}

\usepackage{marvosym}
\usepackage{color}
\usepackage{ulem}

\begin{document}
\title{A Note on Banaszczyk's Inequality}
\author{Hongyuan Qu  \and Chengliang Tian \and  Guangwu Xu}

%\author{Hongyuan Qu\inst{1}\orcidID{0000-0002-1684-0819} \and  Guangwu Xu\inst{1,2,3}\orcidID{0000-0001-6200-3264}$^{(\text{\Letter})}$}

\institute{	}

%\institute{
%	School of Cyber Science and Technology, Shandong University,\\
%	Qingdao, Shandong, 266237, China.\\
%	\email{qhy1qhy@mail.sdu.edu.cn, gxu4sdq@sdu.edu.cn}}
%\author{First Author\inst{1}\orcidID{0000-1111-2222-3333} \and
%Second Author\inst{2,3}\orcidID{1111-2222-3333-4444} \and
%Third Author\inst{3}\orcidID{2222--3333-4444-5555}}
%
\authorrunning{F. Author et al.}
%\institute{Princeton University, Princeton NJ 08544, USA \and
%Springer Heidelberg, Tiergartenstr. 17, 69121 Heidelberg, Germany
\maketitle              % typeset the header of the contribution
\begin{abstract}
Banaszczyk's inequality establishes a tail estimate for the discrete Gaussian measure on a lattice 
in $\mathbb{R}^n$.  This classic result has been influential and plays an important role in lattice-based cryptography. An improvement of the inequality with a transparent proof was given by Tian, Liu and Xu. In this note, 
we further improve this inequality by imposing an appropriate condition, obtaining a significantly better bound. This refined inequality can be used to investigate dual attacks against the Learning With Errors (LWE) problem.
%\keywords{LWE  \and Distinguishing Inequality \and Provable Dual Attack \and Lattice Sparsification \and Modulus Switching }
\end{abstract}
\pagestyle{plain}

\section{Introduction}\label{section_introduction}

Banaszczyk's inequality, introduced by Banaszczyk in 1993 \cite{banaszczyk1993new}, plays an important role in lattice-based cryptography. In essence, this inequality states that after removing the contribution of some short vectors from a  lattice, the remaining Gaussian mass is exponentially small relative  to the total Gaussian mass on the lattice. This inequality was originally used to prove the following 
transference theorem \cite{banaszczyk1993new}:
\begin{theorem}[Transference Theorem \cite{banaszczyk1993new}]
	Let $L\subset\mathbb{R}^n$ be an $n$-dimensional lattice, $L^*$ its dual, and $\lambda_i(L)$ the successive minima of $L$. Then
\[	
		\lambda_i(L)\lambda_{n-i+1}(L^*) \leq n,\ \ \ i = 1,\dots, n.
\]
\end{theorem}
Later, Aharonov and Regev used Banaszczyk's inequality to prove that certain lattice problems belong to $NP \cap coNP$ \cite{aharonov2005lattice}. Their work introduced the idea of using Banaszczyk's inequality for distinguishing distributions. Recently, Pouly and Shen proposed a framework of provable dual attacks on LWE \cite{pouly2024provable}, which was further improved by Qu and Xu \cite{provablelwe}. The core of these provable dual attacks is the use of Banaszczyk's inequality to distinguish between correct and incorrect guesses.

In 2014, Tian, Liu, and Xu  gave a proof of an improved version of Banaszczyk's inequality  \cite{tian2014measure} with a more transparent argument. In this note, we restate the proof of this improvement to make it even more accessible. Furthermore, by imposing an additional natural condition on the lattice, we obtain a further improvement of Banaszczyk's inequality by an exponential factor.

\section{Preliminaries}\label{section_preliminaries}
In this section, we present some necessary knowledge for understanding Banaszczyk's inequality.

For vectors $\bm{x},\bm{y}\in\mathbb{R}^n$, their inner product is written as $\langle \bm{x},\bm{y} \rangle$, and the Euclidean norm of $\bm{x}$ is $\|\bm{x}\|$.  The open ball of radius $r>0$ centered at $\bm{x}$ is defined as $\mathcal{B}(\bm{x},r) = \{ \bm{y}\in\mathbb{R}^n : \|\bm{y}-\bm{x}\| < r \}$.

A lattice in $\mathbb{R}^n$ is a discrete subgroup of the additive group $(\mathbb{R}^n,+)$. 
Discreteness means there exists $\delta>0$ such that $\mathcal{B}(\bm{0},\delta)\cap L = \{ \bm{0} \}$. Every lattice $L\subset\mathbb{R}^n$ has a basis: linearly independent vectors $\bm{b}_1,\dots,\bm{b}_r\in L$ with
\[
L = \{ u_1\bm{b}_1 + \cdots + u_r\bm{b}_r \mid u_1,\dots,u_r \in \mathbb{Z}  \}.
\]
The integer $r$ is the rank of $L$ and $n$ is its dimension. $L$ is called full-rank when $r = n$.  The successive minima of $L$ are
\[
\lambda_i(L) = \min\left\{ t > 0 \mid \dim\left(\mathrm{span}\left\{ \mathbf{v} \in L \mid \|\mathbf{v}\| \leq t \right\}\right) \geq i \right\}, \ \ i = 1,\dots, r.  
\]
The dual lattice of $L$ is
\[
L^* = \{ \bm{u}\in \operatorname{span}(L) : \langle \bm{u},\bm{v} \rangle \in \mathbb{Z} \text{ for all } \bm{v}\in L \},
\]
and one has $(L^*)^* = L$. For $\bm{x}\in\mathbb{R}^n$, its distance to $L$ is
\[
\operatorname{dist}(\bm{x},L) = \min_{\bm{v}\in L} \| \bm{x}-\bm{v}\|.
\]

The Fourier transform of  a rapidly decreasing smooth function\footnote{This means that $h$ and all its (partial) derivatives $D^{\beta}h$ are
	rapidly decreasing in the sense that $\sup_{\bm{x}\in \mathbb{R}^n} |\bm{x}^{\alpha}D^{\beta}h(\bm{x})|<\infty$ for every $\alpha, \beta \in \mathbb{N}^n$. Such a function is said to be in the Schwartz space\cite{stein2011fourier}.} $h:\mathbb{R}^n\rightarrow \mathbb{R}$ is 
\[
\hat{h}(\bm{w}) = \int_{\mathbb{R}^n} h(\bm{x}) e^{-2\pi i\langle \bm{x},\bm{w} \rangle}\,d\bm{x},\qquad \bm{w}\in\mathbb{R}^n.
\]

For $s>0$ and $\bm{x}\in\mathbb{R}^n$, define the Gaussian function $\rho_s(\bm{x}) = e^{-\pi\|\bm{x}\|^2/s^2}$. Its Fourier transform is $\hat{\rho}_s(\bm{y}) = s^n e^{-\pi s^2 \|\bm{y}\|^2} = s^n \rho_{1/s}(\bm{y})$.  
For a discrete set $S\subset\mathbb{R}^n$, set $\rho_s(S) = \sum_{\bm{v}\in S}\rho_s(\bm{v})$.

A key tool is the Poisson summation formula, we state it just for the function $\rho_s(\bm{x})$:
\begin{lemma}[Poisson summation formula \cite{serre2012course}]\label{poisson_summation_formula}
	For an $n$-dimensional lattice $L$, $s>0$, and $\bm{t}\in\mathbb{R}^n$, we have
	\begin{enumerate}
		\item $\rho_s(L) = \dfrac{s^n}{\det(L)}\,\rho_{1/s}(L^*)$.
		\item $\rho_s(L+\bm{t}) = \dfrac{s^n}{\det(L)}\displaystyle\sum_{\bm{w}\in L^*} e^{2\pi i\langle \bm{t},\bm{w} \rangle}\,\rho_{1/s}(\bm{w})$.
	\end{enumerate}
\end{lemma}

\section{Banaszczyk's Inequality: Proof and Improvement}\label{section3}
In this section, we begin with the classic measure inequality of 
Banaszczyk \cite[Lemma 1.5]{banaszczyk1993new}.  The following statement is taken from \cite{tian2014measure} (Theorem 3.1 and its remark), which improves Banaszczyk's inequality by removing the extra factor $2$ that appears in the estimate for cosets (i.e., when $\bm{t}\neq \bm{0}$).
Here we also present a concise proof which makes the argument more accessible.
\begin{lemma}[Banaszczyk's inequality, refined version \cite{banaszczyk1993new,tian2014measure}]\label{lemma2.2}
	Let $L\subset\mathbb{R}^n$ be a lattice of dimension $n$ and $s>0$. Then for any $c\geq 1$ and $\bm{t}\in\mathbb{R}^n$,
	\[
			 \frac{\rho_s\bigg((L+\bm{t}) \setminus \mathcal{B}\bigg(0,cs\sqrt{\frac{n}{2\pi}}\bigg)\bigg)}{\rho_s(L)} \leq \left( c\sqrt{e}\cdot e^{-\frac{c^2}{2}} \right)^n.
	\]
\end{lemma}

\begin{proof}
	\begin{align*}
		\rho_s\bigg((L+&\bm{t}) \setminus \mathcal{B}\bigg(0,cs\sqrt{\frac{n}{2\pi}}\bigg)\bigg)  =\sum_{{\stackrel{\bm{u}\in L+\bm{t}}{\|\bm{u}\|\geq cs\sqrt{\frac{n}{2\pi}}}}}e^{-\frac{\pi}{s^2}\|\bm{u}\|^2}\\
		&=\sum_{\stackrel{\bm{u}\in L+\bm{t}}{\|\bm{u}\|\geq cs\sqrt{\frac{n}{2\pi}}}}e^{-\frac{\pi}{s^2}\left( 1 - \frac{1}{c^2} \right)\|\bm{u}\|^2}\cdot e^{-\frac{\pi}{s^2c^2}\|\bm{u}\|^2}\leq e^{-\frac{\pi}{s^2}\left( 1 - \frac{1}{c^2} \right)\frac{c^2s^2n}{2\pi}}\rho_{sc}(L+\bm{t})\\
		&
		= e^{-\frac{ c^2 - 1 }{2}n}\frac{s^nc^n}{\det(L)}\sum_{\bm{w}\in L^{\ast}}e^{2\pi i\langle \bm{t},\bm{w} \rangle}\rho_{\frac{1}{sc}}(\bm{w})
		\leq e^{-\frac{ c^2 - 1 }{2}n}\frac{s^nc^n}{\det(L)}\sum_{\bm{w}\in L^{\ast}}\rho_{\frac{1}{s}}(\bm{w})\\
		&= e^{-\frac{ c^2 - 1 }{2}n}c^n \rho_s(L).
	\end{align*}
	The result follows since $\rho_s(L) > 0$.
	 \qed
\end{proof}
Note that the Poisson summation formula is used twice in the preceding proof.

As an application of Banaszczyk's inequality, we introduce the following distinguishing inequality, which plays a crucial role in the formal analysis and parameter selection of provable dual attacks \cite{provablelwe}.
\begin{corollary}[Distinguishing inequality{\cite[Corollary 1]{provablelwe}}]\label{corollary_more_tight_bound}
	Let $L\subset \mathbb{R}^n$ be an $n$-dimensional lattice and $\bm{x}\in \mathbb{R}^n$. Let $\tau = {s}\sqrt{\frac{n}{2\pi}}$ and $r_{\bm{x}} = \text{dist}(\bm{x},L)$. 
	If $r_{\bm{x}} \geq \tau$, then
	\begin{equation*}
		\rho_{{s}}(r_{\bm{x}}) \leq \frac{\rho_s(L+\bm{x})}{\rho_s(\bm{x})} \leq \left( \frac{r_{\bm{x}}}{\tau} \right)^n e^{\frac{n}{2}}\rho_{{s}}(r_{\bm{x}}).
	\end{equation*}
	Moreover, the right-hand side is a decreasing function of $r_{\bm{x}}$.
\end{corollary}
\begin{proof}
	Let $\bm{y}\in L+\bm{x}$ that has minimum norm, then $\|\bm{y}\| = r_{\bm{x}}$. 
	
	The left hand side is a well-known fact (\cite{banaszczyk1993new,aharonov2005lattice}) whose proof is straightforward:
	using the fact that $L$ is symmetric about the origin, we  have
	\begin{align*}
		\rho_{s}(L&+\bm{x}) = \frac{1}{2}\sum_{\bm{v}\in L}\left( e^{-\frac{\pi}{s^2}\| \bm{v}+\bm{y} \|^2} + e^{-\frac{\pi}{s^2} \| \bm{v} - \bm{y} \|^2} \right)\\
		&= \frac{1}{2}\sum_{\bm{v}\in L} e^{-\frac{\pi}{s^2}\| \bm{v} \|^2}e^{-\frac{\pi}{s^2}\| \bm{y} \|^2}\left( e^{-\frac{2\pi}{s^2}\langle \bm{v},\bm{y} \rangle} + e^{\frac{2\pi}{s^2}\langle \bm{v},\bm{y} \rangle}  \right)\geq \rho_s(L)\rho_s(r_{\bm{x}}).
	\end{align*}

For the right hand side, write $r_{\bm{x}} = {c}{s}\sqrt{\frac{n}{2\pi}}$. Then $r_{\bm{x}} \geq \tau$ is equivalent to $c \geq 1$. Noting that $c = \frac{r_{\bm{x}}\sqrt{2\pi}}{s\sqrt{n}} = \frac{r_{\bm{x}}}{\tau}$, by Lemma \ref{lemma2.2} we have
	\begin{align*}
		\frac{\rho_s(L+\bm{x})}{\rho_s(L)} \leq \left( c\sqrt{e}\cdot e^{-\frac{c^2}{2}} \right)^n
		= \left( \frac{r_{\bm{x}}\sqrt{e}}{\tau}\cdot e^{-\frac{\pi r_{\bm{x}}^2}{s^2n}} \right)^n
		= \left( \frac{r_{\bm{x}}}{\tau} \right)^n e^{\frac{n}{2}}\rho_{{s}}(r_{\bm{x}}).
	\end{align*}
	For $c > 1$, the expression $c\sqrt{ e}\cdot e^{-\frac{c^2}{2}}$ is monotonically decreasing in $c$, while $r_{\bm{x}}$ increases with  $c$. Consequently, the composite term $\left( \frac{r_{\bm{x}}}{\tau} \right)^n e^{\frac{n}{2}}\rho_{{s}}(r_{\bm{x}})$ is monotonically decreasing in $r_{\bm{x}}$. 
	\qed
\end{proof}

It is natural to take  the length of the shortest non-zero vector in our consideration. We impose an additional condition that $\lambda_1(L)$ is suitably large and obtain the following improved inequality:
\begin{lemma}[Improved Banaszczyk Inequality]\label{lemma3.1}
	Let $L\subset\mathbb{R}^n$ be a lattice. Assume that $\lambda_1(L) \geq kcs\sqrt{\frac{n}{2\pi}}$ where $k > 1$. Then for any $c\geq 1$ and $\bm{t}\in\mathbb{R}^n$,
	\begin{equation*}
		{\rho_s\left((L+\bm{t})\setminus \mathcal{B}\left(\bm{0},cs\sqrt{\frac{n}{2\pi}}\right)\right)} \leq \frac{\left(  e^{1 - c^2} \right)^{\frac{n}{2}}}{1 - \epsilon},
	\end{equation*}
	where $\epsilon = \left( k\sqrt{e}\cdot e^{-\frac{k^2}{2}} \right)^n$.
\end{lemma}
\begin{proof}
	We have
	\begin{align*}
		\rho_s\bigg((L&+\bm{t})\setminus \mathcal{B}\left(\bm{0},cs\sqrt{\frac{n}{2\pi}}\right)\bigg) = \sum_{\substack{\bm{v}\in L+\bm{t}\\ \|\bm{v}\| \ge cs\sqrt{\frac{n}{2\pi}}}}e^{-\frac{\pi}{s^2}\left( 1- \frac{1}{c^2} \right)\|\bm{v}\|^2}\cdot e^{-\frac{\pi}{s^2c^2}\|\bm{v}\|^2}\\
		&\leq e^{-\frac{\pi}{s^2}\left( 1- \frac{1}{c^2} \right)\frac{c^2s^2n}{2\pi}}\cdot \rho_{sc}(L+\bm{t})\leq \left(  e^{1 - c^2} \right)^{\frac{n}{2}}\cdot \rho_{sc}(L).
	\end{align*}
	Since $\lambda_1(L) \geq kcs\sqrt{\frac{n}{2\pi}}$, using Lemma \ref{lemma2.2} we obtain
	\begin{align*}
		\frac{\rho_{sc}(L) - 1 }{\rho_{sc}(L)} = \frac{\rho_{sc}(L\backslash\mathcal{B}(\bm{0}, kcs\sqrt{\frac{n}{2\pi}}))}{\rho_{sc}(L)} \leq \left( k\sqrt{e}\cdot e^{-\frac{k^2}{2}} \right)^n.
	\end{align*}
	Let $\epsilon =\left( k\sqrt{e}\cdot e^{-\frac{k^2}{2}} \right)^n$. For $k>1$, $\epsilon < 1$. Then $\rho_{sc}(L) \leq \frac{1}{1-\epsilon}$.Thus we have
	\begin{equation*}
		\rho_s\left((L+\bm{t})\setminus \mathcal{B}\left(\bm{0},cs\sqrt{\frac{n}{2\pi}}\right)\right) \leq \left(  e^{1 - c^2} \right)^{\frac{n}{2}}\cdot \rho_{sc}(L) \leq \frac{\left(  e^{1 - c^2} \right)^{\frac{n}{2}}}{1 - \epsilon}.
	\end{equation*}\qed
\end{proof}
Since $\rho_s(L) > 1$ for any $s>0$, the following is immediate by Lemma~\ref{lemma3.1}:
\begin{corollary}\label{corollary_optimized_tight_bound}
	Let $L\subset \mathbb{R}^n$ and $s>0$. Let $\bm{x}\in \mathbb{R}^n$ and $r_{\bm{x}} = \mbox{dist}(\bm{x},L)$.  Assume $r_{\bm{x}} \geq {cs}\sqrt{\frac{n}{2\pi}}$ and $\lambda_1(L)  \geq {kcs}\sqrt{\frac{n}{2\pi}}$ for $c\geq 1$ and $k>1$, then
	\begin{equation*}
		\rho_{{s}}(r_{\bm{x}}) \leq \frac{\rho_s(L+\bm{x})}{\rho_s(L)} \leq \frac{\left( e^{1-c^2} \right)^{\frac{n}{2}}}{1-\epsilon},
	\end{equation*}
	where $\epsilon = \left( k\sqrt{e}\cdot e^{-\frac{k^2}{2}} \right)^n$.
\end{corollary}
%\begin{proof}
%	The left inequality follows directly from Corollary \ref{corollary_more_tight_bound}. Since $r_{\bm{x}} \geq cr$ where $c\geq 1$, we have $\rho_s(L+\bm{x}) = \rho_s\left(L+\bm{x}\backslash \mathcal{B}\left(\bm{0},cs\sqrt{\frac{n}{2\pi}}\right)\right)$. By Lemma \ref{lemma3.1} we obtain
%\begin{align*}
%       \frac{\rho_s(L+\bm{x})}{\rho_s(L)} &= \frac{\rho_s\left(L+\bm{x}\backslash \mathcal{B}\left(\bm{0},cs\sqrt{\frac{n}{2\pi}}\right)\right)}{\rho_s(L)} \leq \frac{\left(  e^{1-c^2} \right)^{\frac{n}{2}}}{1-\epsilon}.
%   \end{align*}
%   \qed
%\end{proof}
In a typical application, $n\ge 500$ (e.g., \cite{pouly2024provable,provablelwe}). The distinguishing inequality can be improved to
\begin{corollary}[Improved distinguishing inequality]\label{corollary_optimized_tight_bound_specific}
Let $L\subset \mathbb{R}^n$ and $s>0$. Let $\bm{x}\in \mathbb{R}^n$ and $r_{\bm{x}} = \mbox{dist}(\bm{x},L)$. Assume $r_{\bm{x}} \geq {cs}\sqrt{\frac{n}{2\pi}}$ and $\lambda_1(L) \geq {kcs}\sqrt{\frac{n}{2\pi}}$ where $c\geq 1$ and $k\geq 1.04$. Assume $n\geq 500$, then
\begin{equation*}
	\rho_{{s}}(r_{\bm{x}}) \leq \frac{\rho_s(L+\bm{x})}{\rho_s(L)} \leq 2\left(e^{1-c^2}\right)^{\frac{n}{2}}.
\end{equation*}
\end{corollary}
\begin{proof}
When $k\geq 1.04$ and $n\geq 500$, we have
\begin{equation*}
	\left( k\sqrt{e}\cdot e^{-\frac{k^2}{2}} \right)^n < 0.5.
\end{equation*}
Then the result follows from Corollary \ref{corollary_optimized_tight_bound}. \qed
\end{proof}

\begin{remark}
Note that in Corollary \ref{corollary_optimized_tight_bound_specific}, the term $2\left(e^{1-c^2}\right)^{\frac{n}{2}}$ can be equivalently  expressed as $2e^{\frac{n}{2}}e^{\frac{\pi}{s^2}(r_{\bm{x}}^2 - r^2c^2)}\rho_s(r_{\bm{x}})$. Hence, the above improved distinguishing inequality can be written equivalently in the following unified form, where the same factor $\rho_s(r_{\bm{x}})$ appears on both sides:
\begin{equation*}
	\rho_{{s}}(r_{\bm{x}}) \leq \frac{\rho_s(L+\bm{x})}{\rho_s(L)} \leq 2e^{\frac{n}{2}}e^{\frac{\pi}{s^2}(r_{\bm{x}}^2 - r^2c^2)}\rho_s(r_{\bm{x}}).
\end{equation*}
\end{remark}

Observe that
\begin{equation*}
\frac{\left( c\sqrt{ e}\cdot e^{- \frac{c^2}{2}} \right)^n}{2\left( e^{1-c^2} \right)^{\frac{n}{2}}} = \frac{c^n}{2}.
\end{equation*}
Here, the numerator is the upper bound given by Banaszczyk inequality. For $c>1$, we obtain an improvement of Banaszczyk inequality by an exponential factor.

\begin{credits}

%\subsubsection{\ackname} We thank the anonymous reviewers for valuable comments
%and suggestions.	This work is partially supported by the National Natural Science Foundation of China (No. 12271306) .

%\subsubsection{\discintname}
%There are no financial conflicts of interest to disclose.
\end{credits}

 \bibliographystyle{splncs04}
 \bibliography{mybibliography}

@InProceedings{pouly2024provable,
author="Pouly, Amaury
and Shen, Yixin",
editor="Joye, Marc
and Leander, Gregor",
title="{Provable Dual Attacks on Learning with Errors}",
booktitle="Advances in Cryptology -- EUROCRYPT 2024",
year="2024",
publisher="Springer Nature Switzerland",
address="Cham",
pages="256--285",
isbn="978-3-031-58754-2",
url={https://doi.org/10.1007/978-3-031-58754-2\_10}
}

@article{aharonov2005lattice,
author = {Aharonov, Dorit and Regev, Oded},
title = {{Lattice problems in NP $\cap$ coNP}},
year = {2005},
issue_date = {September 2005},
publisher = {Association for Computing Machinery},
address = {New York, NY, USA},
volume = {52},
number = {5},
issn = {0004-5411},
url = {https://doi.org/10.1145/1089023.1089025},
journal = {J. ACM},
month = sep,
pages = {749-765},
numpages = {17}
}

@book{stein2011fourier,
  title={{Fourier analysis: An Introduction}},
  author={Stein, Elias M and Shakarchi, Rami},
  volume={1},
  year={2003},
  publisher={Princeton University Press}
}

@article{banaszczyk1993new,
  title={{New bounds in some transference theorems in the geometry of numbers}},
  author={Banaszczyk, Wojciech},
  journal={Mathematische Annalen},
  volume={296},
  pages={625--635},
  year={1993},
  publisher={Springer-Verlag},
}

@article{tian2014measure,
  title={{Measure inequalities and the transference theorem in the geometry of numbers}},
  author={Tian, Chengliang and Liu, Mingjie and Xu, Guangwu},
  journal={Proceedings of the American Mathematical Society},
  volume={142},
  number={1},
  pages={47--57},
  year={2014}
}

@book{serre2012course,
	title={{A Course in Arithmetic}},
	author={Serre, Jean-Pierre},
	volume={7},
	year={2012},
	publisher={Springer Science \& Business Media}
}

@misc{provablelwe,
      author = {Hongyuan Qu and Guangwu Xu},
      title = {{On the Provable Dual Attack for {LWE} by Modulus Switching}},
      howpublished = {Cryptology {ePrint} Archive, Paper 2025/859},
      year = {2025},
      url = {https://eprint.iacr.org/2025/859}
}
%Author, F.: Article title. Journal \textbf{2}(5), 99--110 (2016)

%Author, F., Author, S.: Title of a proceedings paper. In: Editor,
%F., Editor, S. (eds.) CONFERENCE 2016, LNCS, vol. 9999, pp. 1--13.
%Springer, Heidelberg (2016). \doi{10.10007/1234567890}

%Author, F., Author, S., Author, T.: Book title. 2nd edn. Publisher,

%Author, A.-B.: Contribution title. In: 9th International Proceedings
%on Proceedings, pp. 1--2. Publisher, Location (2010)

%LNCS Homepage, \url{http://www.springer.com/lncs}, last accessed 2023/10/25
%\end{thebibliography}

\begin{appendix}
\end{appendix}

\end{document}